\documentclass[11pt,letterpaper]{article}

\usepackage{fullpage}
\usepackage{amsfonts,amsthm,amsmath}
\usepackage{url,hyperref}
\usepackage{graphicx,float}
\usepackage{physics,mathtools}
\usepackage{color}
\usepackage{xspace}

\usepackage{algorithm, algorithmicx}
\usepackage[noend]{algpseudocode}

\newtheorem{theorem}{Theorem}[section]
\newtheorem{lemma}[theorem]{Lemma}
\newtheorem{proposition}[theorem]{Proposition}

\newtheorem{definition}[theorem]{Definition}

\definecolor{colorx}{rgb}{0.9, 0.1, 0.3}

\newcommand{\learn}{\textsc{Learn} }
\newcommand{\idle}{\textsc{Idle} }

\newcommand{\follow}{\textsc{Follow} }

\newcommand{\Brac}[1]{\left[#1\right]}
\newcommand{\Ex}[1]{\mathbb{E}\Brac{#1}}

\newcommand{\eat}[1]{}

\usepackage{xcolor}

\newcommand{\ALG}{\mathsf{ALG}\xspace}
\newcommand{\OPT}{\mathsf{OPT}\xspace}
\newcommand{\cost}{\mathsf{cost}\xspace}

\newcommand{\static}{\textsc{Static}\xspace}

\newcommand{\prev}{\mathsf{prev}\xspace}
\newcommand{\next}{\mathsf{next}\xspace}
\newcommand{\predictnext}{\mathsf{pnext}\xspace}

\title{Online Algorithms for Weighted Paging with Predictions\thanks{A preliminary version of this paper was published in the Proceedings of the 47th {\em International Colloquium on Automata, Languages and Programming} (ICALP), 2020.}}

\author{Zhihao Jiang\,\thanks{Institute for Interdisciplinary Information Sciences, Tsinghua University, Beijing, China. Work done while visiting Duke University, Durham, NC, USA. Email: {\tt jzh16@mails.tsinghua.edu.cn}.}
\and
Debmalya Panigrahi\thanks{Department of Computer Science, Duke University, Durham, NC, USA. This work was supported in part by NSF grants CCF-1535972, CCF-1955703, an NSF CAREER Award CCF-1750140, and the Indo-US Virtual Networked Joint Center on Algorithms under Uncertainty. Email: {\tt debmalya@cs.duke.edu}.}
\and
Kevin Sun\thanks{Department of Computer Science, Duke University, Durham, NC, USA. This work was supported in part by NSF grants CCF-1535972, CCF-1955703, and an NSF CAREER Award CCF-1750140. Email: {\tt ksun@cs.duke.edu}.}
}

\date{}

\begin{document}

\maketitle

\begin{abstract}
    In this paper, we initiate the study of the weighted paging problem with predictions. This continues the recent line of work in online algorithms with predictions, particularly that of Lykouris and Vassilvitski (ICML 2018) and Rohatgi (SODA 2020) on unweighted paging with predictions. We show that unlike unweighted paging, neither a fixed lookahead nor knowledge of the next request for every page is sufficient information for an algorithm to overcome existing lower bounds in weighted paging. However, a combination of the two, which we call the strong per request prediction (SPRP) model, suffices to give a 2-competitive algorithm. We also explore the question of gracefully degrading algorithms with increasing prediction error, and give both upper and lower bounds for a set of natural measures of prediction error. 
\end{abstract}

\clearpage

\section{Introduction}
\label{sec:introduction}
The paging problem is among the most well-studied problems in online algorithms. In this problem, there is a set of $n$ pages and a cache of size $k < n$. The online input comprises a sequence of requests for these pages. If the requested page is already in the cache, then the algorithm does not need to do anything. But, if the requested page is not in the cache, then the algorithm suffers what is known as a {\em cache miss} and must bring the requested page into the cache. If the cache is full, then an existing page must be evicted from the cache to make room for the new page. The goal of the online algorithm is to minimize the total number of cache misses in the {\em unweighted paging} problem, and the total weight of the evicted pages in the {\em weighted paging} problem. It is well-known that for both problems, the best deterministic algorithms have a competitive ratio of $O(k)$ and the best randomized algorithms have a competitive ratio of $O(\log k)$ (see, e.g., \cite{chrobak1991new,bansal2012primal}).
 
Although the paging problem is essentially solved from the perspective of competitive analysis, it also highlights the limitations of this framework. For instance, it fails to distinguish between algorithms that perform nearly optimally in practice such as the {\em least recently used} (LRU) rule and very na\"ive strategies such as {\em flush when full} that evicts all pages whenever the cache is full. In practice, paging algorithms are augmented with predictions about the future (such as those generated by machine learning models) to improve their empirical performance. To model this, for unweighted paging, several lookahead models have been proposed where only a partial prediction of the future leads to algorithms that are significantly better than what can be obtained in traditional competitive analysis. But, to the best of our knowledge, no such results were previously known for the weighted paging problem. In this paper, we initiate the study of the {\em weighted paging problem with future predictions}.  

For unweighted paging, it is well-known that evicting the page whose next request is {\em farthest in the future} (also called Belady's rule) is optimal. As a consequence, it suffices for an online algorithm to simply predict the next request of every page (we call this {\em per request prediction} or PRP in short) in order to match offline performance. In fact, Lykouris and Vassilvitskii~\cite{lykouris2018competitive} (see also Rohatgi~\cite{rohatgi2020near}) showed recently that in this prediction model, one can simultaneously achieve a competitive ratio of $O(1)$ if the predictions are accurate, and $O(\log k)$ regardless of the quality of the predictions. Earlier, Albers~\cite{albers1993influence} used a different prediction model called {\em $\ell$-strong lookahead}, where we predict a sequence of future requests that includes $\ell$ distinct pages (excluding the current request). For $\ell = n-1$, this prediction is stronger than the PRP model, since the algorithm can possibly see multiple requests for a page in the lookahead sequence. But, for $\ell < n-1$, which is typically the setting that this model is studied in, the two models are incomparable. The main result in \cite{albers1993influence} is to show that one can obtain a constant approximation for unweighted paging for $\ell \geq k-2$. 
 
Somewhat surprisingly, we show that neither of these models are sufficient for weighted paging. In particular, we show a lower bound of $\Omega(k)$ for deterministic algorithms and $\Omega(\log k)$ for randomized algorithms in the PRP model. These lower bounds match, up to constants, standard lower bounds for the online paging problem (without prediction) (see, e.g., \cite{MotwaniR}), hence establishing that the PRP model does not give any advantage to the online algorithm beyond the strict online setting. Next, we show that for $\ell$-strong lookahead, even with $\ell = k$, there are lower bounds of $\Omega(k)$ for deterministic algorithms and $\Omega(\log k)$ for randomized algorithms, again asymptotically matching the lower bounds from online paging without prediction. Interestingly, however, we show that a combination of these prediction models is sufficient: if $\ell = n-1$ in the strong lookahead setting, then we get predictions that subsume both models; and, in this case, we give a simple deterministic algorithm with a competitive ratio of $2$ for weighted paging, thereby overcoming the online lower bounds.
 
Obtaining online algorithms with predictions, however, is fraught with the risk that the predictions are inaccurate which renders the analysis of the algorithms useless. Ideally, one would therefore, want the algorithms to also be robust, in that their performance gracefully degrades with increasing prediction error. Recently, there has been significant interest in designing online algorithms with predictions that achieve both these goals, of matching nearly offline performance if the predictions are correct, and of gracefully degrading as the prediction error increases. Originally proposed for the (unweighted) paging problem~\cite{lykouris2018competitive}, this model has gained significant traction in the last couple of years and has been applied to problems in data structures~\cite{mitzenmacher2018model}, online decision making~\cite{purohit2018improving,gollapudi2019online}, scheduling theory~\cite{purohit2018improving,LattanziLMV20}, frequency estimation~\cite{hsu2018learningbased}, etc. Our final result contributes to this line of research.

First, if the online algorithm and offline optimal solution both use a cache of size $k$, then we show that no algorithm can asymptotically benefit from the predictions while achieving sublinear dependence on the prediction error. Moreover, if we make the relatively modest assumption that the algorithm is allowed a cache that contains just $1$ extra slot than that of the optimal solution, then we can achieve constant competitive ratio when the prediction error is small.

\subsection{Overview of models and our results}
Our first result is a lower bound for weighted paging in the PRP model. Recall that in the PRP model, in addition to the current page request, the online algorithm is provided the time-step for the next request of the same page. For instance, if the request sequence is $(a, b, a, c, d, b, \ldots)$, then at time-step $1$, the algorithm sees request $a$ and is given position 3, and at time-step 2, the algorithm sees request $b$ and is given position 6.

\begin{theorem} \label{thm:1-prp}
For weighted paging with PRP, any deterministic algorithm is $\Omega(k)$-competitive, and any randomized algorithm is $\Omega(\log k)$-competitive.
\end{theorem}
Note that these bounds are tight, because there exist online algorithms without prediction whose competitive ratios match these bounds (see Chrobak {\em et al.}~\cite{chrobak1991new} and Bansal {\em et al.}~\cite{bansal2012primal}).

Next, for the $\ell$-strong lookahead model, we show lower bounds for weighted paging. Recall that in this model, the algorithm is provided a lookahead into future requests that includes $\ell$ distinct pages. For instance, if $\ell = 3$ and the request sequence is $(a, b, a, c, d, b, \ldots)$, then at time-step 1, the algorithm sees request $a$ and is given the lookahead sequence $(b,a,c)$ since it includes 3 distinct pages. At time step 2, the algorithm sees request $b$ and is given $(a,c,d)$. Note the difference with the PRP model, which would not be give the information that the request in time-step $5$ is for page $d$, but does give the information that the request in time-step $6$ is for page $b$.

\begin{theorem} \label{thm:1-ell-small}\label{thm:1-ell-large}
For weighted paging with $\ell$-strong lookahead where $\ell\leq n-k$, any deterministic algorithm is $\Omega(k)$-competitive, and any randomized algorithm is $\Omega(\log k)$-competitive.

\smallskip
For weighted paging with $\ell$-strong lookahead where $n-k+1\leq \ell\leq n-1$, any deterministic algorithm is $\Omega(n-\ell)$-competitive, and any randomized algorithm is $\Omega(\log (n-\ell))$-competitive.
\end{theorem}

In contrast to these lower bounds, we show that a prediction model that combines features of these individual models gives significant benefits to an online algorithm. In particular, combining PRP and $\ell$-strong lookahead, we define the following prediction model:

\begin{center}
\fbox{
\begin{minipage}{0.98\linewidth}
{\bf SPRP (``strong per-request prediction''):} On a request for page $p$, the predictor gives the next time-step when $p$ will be requested {\em and all page requests till that request}.
\end{minipage}}
\end{center}

This is similar to $(n-1)$-strong lookahead, but is slightly weaker in that it does not provide the first request of every page at the outset. After each of the $n$ pages has been requested, SPRP and $(n-1)$-strong lookahead are equivalent.

\begin{theorem} \label{thm:1-static}
There is a deterministic 2-competitive for weighted paging with SPRP.
\end{theorem}

So far, all of these results assume that the prediction model is completely correct. However, in general, predictions can have errors, and therefore, it is desirable that an algorithm gracefully degrades with increase in prediction error. To this end, we also give upper and lower bounds in terms of the prediction error. 

For unweighted paging, Lykouris and Vassilvitski~\cite{lykouris2018competitive} basically considered two measures of prediction error. The first, called $\ell_{pd}$ in this paper, is defined as follows: For each input request $p_t$, we increase $\ell_{pd}$ by $w(p_t)$ times the absolute difference between the predicted next-arrival time and the actual next-arrival time. For unweighted paging, Lykouris and Vassilvitskii~\cite{lykouris2018competitive} gave an algorithm with cost $O(\OPT+\sqrt{\ell_{pd}\cdot \OPT})$. Unfortunately, we rule out an analogous result for weighted paging.
\begin{theorem} \label{thm:1-lpd-lower}
For weighted paging with SPRP, there is no deterministic algorithm whose cost is $o(k) \cdot \OPT + o(\ell_{pd})$, and there is no randomized algorithm whose cost is $o(\log k)\cdot \OPT + o(\ell_{pd})$.
\end{theorem}

It turns out that the $\ell_{pd}$ error measure is closely related to another natural error measure that we call the $\ell_1$ measure. This is defined as follows: for each input request $p_t$, if the prediction $q_t$ is not the same as $p_t$, then increase $\ell_1$ by the sum of weights $w(p_t) + w(q_t)$. (This is the $\ell_1$ distance between the predictions and actual requests in the standard weighted star metric space for the weighted paging problem.) The lower bound for $\ell_{pd}$ continues to hold for $\ell_1$ as well, and is tight.

\begin{theorem} \label{thm:1-l1-lower} \label{thm:1-l1-upper}
For weighted paging with SPRP, there is no deterministic algorithm whose cost is $o(k) \cdot \OPT + o(\ell_{1})$, and there is no randomized algorithm whose cost is $o(\log k)\cdot \OPT + o(\ell_{1})$.
Furthermore, there is a deterministic algorithm with SPRP with cost $O (\OPT + \ell_1)$.
\end{theorem}

One criticism of both the $\ell_{pd}$ and $\ell_1$ error measures is that they are not robust to insertions or deletions from the prediction stream. To counter this, Lykouris and 
Vassilvitski~\cite{lykouris2018competitive} used a variant of the classic edit distance measure, and showed a constant competitive ratio for this error measure. For weighted paging, we also consider a variant of edit distance, called $\ell_{ed}$ and formally defined in Section~\ref{sec:sprp}, which allows insertions and deletions between the predicted and actual request streams.\footnote{For technical reasons, neither $\ell_{ed}$ in this paper nor the edit distance variant in \cite{lykouris2018competitive} exactly match the classical definition of edit distance.} Unfortunately, as with $\ell_{pd}$ and $\ell_1$, we rule out algorithms that asymptoticaly benefit from the predictions while achieving sublinear dependence on $\ell_{ed}$.
Furthermore, if the algorithm were to use a cache with even one extra slot than the optimal solution, then we show that even for weighted paging, we can achieve a constant competitive algorithm. We summarize these results in the next theorem.

\begin{theorem} \label{thm:1-led-lower} \label{thm:1-led-upper}
For weighted paging with SPRP, there is no deterministic algorithm whose cost is $o(k) \cdot \OPT + o(\ell_{ed})$, and there is no randomized algorithm whose cost is $o(\log k)\cdot \OPT + o(\ell_{ed})$.

\smallskip\noindent
In the same setting, there exists a randomized algorithm that uses a cache of size $k+1$ whose cost is $ O (\OPT + \ell_{ed})$, where $\OPT$ uses a cache of size $k$.
\end{theorem}

\subsection{Related work}
We now give a brief overview of the online paging literature, highlighting the results that consider a prediction model for future requests. For unweighted paging, the optimal offline algorithm is Belady's algorithm, which always evicts the page that appears farthest in the future~\cite{belady1966study}. For online paging, Sleator and Tarjan~\cite{sleator1985amortized} gave a deterministic $k$-competitive algorithm, and Fiat et al.~\cite{fiat1991competitive} gave a randomized $O(\log k)$-competitive algorithm; both results were also shown to be optimal. For weighted online paging, Chrobak {\em et al.}~\cite{chrobak1991new} gave a deterministic $k$-competitive algorithm, and Bansal {\em et al.}~\cite{bansal2012primal} gave an $O(\log k)$-competitive randomized algorithm, which are also optimal by extension.

Recently, Lykouris and Vassilvitskii~\cite{lykouris2018competitive} introduced a prediction model that we call PRP in this paper: on each request $p$, the algorithm is given a prediction of the next time at which $p$ will be requested. For unweighted paging, they gave a randomized algorithm, based on the ``marker'' algorithm of Fiat et al.~\cite{fiat1991competitive}, with competitive ratio $O(\min(\sqrt{\ell_{pd}/\OPT}, \log k))$. Here, $\ell_{pd}$ is the absolute difference between the predicted arrival and actual arrival times of requests, summed across all requests. They also perform a tighter analysis yielding a competitive ratio of $O(\min(\eta_{ed}/\OPT, \log k))$, where $\eta_{ed}$ is the edit distance between the predicted sequence and the actual input. Subsequently, Rohatgi~\cite{rohatgi2020near} improved the former bound to $O(1+\min((\ell_{pd}/\OPT)/k,1)\log k)$ and also proved a lower bound of $\Omega(\log \min((\ell_{pd}/\OPT)/(k\log k),k))$.

Albers~\cite{albers1993influence} studied the $\ell$-strong lookahead model: on each request $p$, the algorithm is shown the next $\ell$ distinct requests after $p$ and all pages within this range. For unweighted paging, Albers~\cite{albers1993influence} gave a deterministic $(k-\ell)$-competitive algorithm and a randomized $2H_{k-\ell}$-competitive algorithm. Albers also showed that these bounds are essentially tight: if $l\leq k-2$, then any deterministic algorithm has competitive ratio at least $k-\ell$, and any randomized algorithm has competitive ratio at least $\Omega(\log(k-\ell))$.

Finally, we review the paging model in which the offline adversary is restricted to a cache of size $h < k$, while the online algorithm uses a larger cache of size $k$. For this model, Young~\cite{young2002line} gave a deterministic algorithm with competitive ratio $k/(k-h+1)$ and showed that this is optimal. In another paper, Young~\cite{young1991online} showed that the randomized ``marker'' algorithm is $O(\log (k/k-h))$-competitive and this bound is optimal up to constants.

\paragraph*{Remark}
The independent, concurrent work of Antoniadas et al.~\cite{antoniadis2020online} has slight overlap with ours. In particular, they also showed that the PRP prediction model does not provide asymptotic benefits for randomized algorithms. They also gave a prediction-based randomized algorithm for \emph{unweighted} caching, and they note that their prediction error is not directly comparable to the error used by Lykouris and Vassilvitskii~\cite{lykouris2018competitive} and Rohatgi~\cite{rohatgi2020near}.

\paragraph*{Roadmap}
In Section~\ref{sec:PRP}, we show the lower bounds stated in Theorem~\ref{thm:1-prp} for the PRP model. The lower bounds for the $\ell$-strong lookahead model stated in Theorem~\ref{thm:1-ell-small} are proven in Section~\ref{sec:l-strong}. In Section~\ref{sec:static}, we state and analyze the algorithm for the SPRP model with no error, thereby proving Theorem~\ref{thm:1-static}. Finally, in Section~\ref{sec:sprp}, we consider the SPRP model with errors, and focus on the upper and lower bounds in Theorems~\ref{thm:1-lpd-lower},~\ref{thm:1-l1-lower}, and~\ref{thm:1-led-lower}.

\section{The Per-Request Prediction Model (PRP)} \label{sec:PRP}
In this section, we give the lower bounds stated in Theorem~\ref{thm:1-prp} for the PRP model. Our strategy, at a high level, will be the same in both the deterministic and randomized cases: we consider the special case where the cache size is exactly one less than the number of distinct pages. We then provide an algorithm that generates a specific input. In the deterministic case, this input will be adversarial, based on the single page not being in the cache at any time. In the randomized case, the input will be oblivious to the choices made by the paging algorithm but will be drawn from a distribution. We will give a brief overview of the main ideas that are common to both lower bound constructions first, and then give the details of the randomized construction in this section.

Let us first recall the $\Omega(k)$ deterministic lower bound for unweighted caching without predictions. Suppose the cache has size $k$ and the set of distinct pages is $\{a_0, a_1, \ldots, a_k\}$. At each step, the adversary requests the page $a_\ell$ not contained in the cache of the algorithm $\ALG$. Then $\ALG$ incurs a miss at every step, while $\OPT$, upon a miss, evicts the page whose next request is furthest in the future. Therefore, $\ALG$ misses at least $k$ more times before $\OPT$ misses again. 

Ideally, we would like to imitate this construction. But, the adversary cannot simply request the missing page $a_\ell$ because that could violate the predictions made on previous requests. Our first idea is to replace this single request for $a_\ell$ with a ``block'' of requests of pages {\em containing} $a_\ell$
in a manner that all the previous predictions are met, but $\ALG$ still incurs the cost of page $a_\ell$ in serving this block of requests.

But, how do we guarantee that $\OPT$ only misses requests once for every $k$ blocks? Indeed, it is not possible to provide such a guarantee. Instead, as a surrogate for $\OPT$, we use an array of $k$ algorithms $\ALG_i$ for $1\leq i\leq k$, where each $\ALG_i$ follows a fixed strategy: maintain all pages except $a_0$ and $a_i$ permanently in the cache, and swap $a_0$ and $a_i$ as required to serve their requests. Our goal is to show that the sum of costs of all these algorithms is a lower bound (up to constants) on the cost of $\ALG$; this would clearly imply an $\Omega(k)$ lower bound. 

This is where the weights of pages come handy. We set the weight $w(a_i)$ of page $a_i$ in the following manner: $w(a_i) = c^i$ for some constant $c \geq 2$. Now, imagine that a block requested for a missing page $a_\ell$ only contains pages $a_0, a_1, \ldots, a_\ell$ (we call this an $\ell$-block). The algorithms $\ALG_i$ for $i \leq \ell$ suffer a cache miss on page $a_i$ in this block, while the remaining algorithms $\ALG_i$ for $i > \ell$ do not suffer a cache miss in this block. Moreover, the sum of costs of all the algorithms $\ALG_i$ for $i \leq \ell$ in this block is at most a constant times that of the cost of $\ALG$ alone, because of the geometric nature of the cost function. 

The only difficulty is that by constructing blocks that do not contain pages $a_i$ for $i > \ell$, we might be violating the previous predictions for these pages. To overcome this, we create an invariant where for every $i$, an $(i+1)$-block must be introduced after a fixed number of $i$-blocks. Because of this invariant, we are sometimes forced to introduce a larger block than that demanded by the missing page in $\ALG$. To distinguish between these two types of blocks, we call the ones that exactly correspond to the missing page a {\em regular} block, and the ones that are larger {\em irregular} blocks. Irregular blocks help preserve the correctness of all previous predictions, but the sum of costs of $\ALG_i$'s on an irregular block can no longer be bounded against that of $\ALG$. Nevertheless, we can show that the number of irregular blocks is small enough that this extra cost incurred by $\ALG_i$'s in irregular blocks can be charged off to the regular blocks, thereby proving the deterministic lower bound. The randomized lower bound follows the same intuition.

\subsection{Deterministic Lower Bound} \label{subsec:prp-det}
Now we give a formal proof of the following theorem.

\begin{theorem} \label{thm:prp-det}
For weighted paging with PRP, any deterministic algorithm is $\Omega(k)$-competitive.
\end{theorem}

For simplicity, we assume that all algorithms start with an empty cache. While generating the input sequence, we will maintain variables $u_i$ and $t$ that satisfy the following invariants:
\begin{itemize}
    \item The value of $u_i$ denotes the next time at which page $a_i$ will arrive.
    \item The value of $t$ is the number of requests that have been made, initialized to $t=0$.
\end{itemize}
The input is defined as follows:
\begin{enumerate}
    \item For $0\leq i\leq k$, let $u_i=(2c+2)^i$, and for $0 \leq i < k$, let $y_i=0$.
    \item Repeat the following:\label{prp1:repeat}
    \begin{enumerate}
        \item Let $\ell$ denote the largest index such that $a_{\ell}$ is not in the cache. 
        \item Increase $\ell$ until $\ell = k$ or $y_\ell < 2c$. \label{prp1:inc_i}
        \item For $j$ from $0$ to $\ell$, \label{prp1:for}
        \begin{enumerate}
            \item Set all the requests from time $t+1$ through $u_j-1$ as $a_{j-1}$. (Note: If $j = 0$, then $u_j = t+1$, so this step is empty.)
            \item Set the request at time $u_j$ to be $a_j$.
            \item Let $t = u_j$.
        \end{enumerate}
        \item For $0\leq j\leq \ell$, let $u_j=t+(2c+2)^j$. \label{prp1:uj}
        \item For $0\leq j<\ell$, let $y_j=0$. If $\ell<k$, increase $y_{\ell}$ by one.
    \end{enumerate}
\end{enumerate}

We call the requests generated each time we enter Step (\ref{prp1:repeat}) a \emph{block}; if the final value of $\ell$ is $i$ then this is an \emph{$i$-block}.
Let us give an overview of the lower bound argument. Firstly, we show that every $i$-block is a contiguous sequence of $a_0$'s, then $a_1$'s, and so on, ending with a single $a_i$ (Lemma~\ref{lem:det-contiguous}). Thus, for each such block, $\ALG$ incurs a cost of at least $c^i$, because at the beginning of this block, the cache of $\ALG$ does not contain the page $a_i$.

On the other hand, for each $i \in \{1, 2, \ldots, k\}$, consider the algorithm $\ALG_i$ defined as follows: upon a cache miss, evict $a_i$ if it is in the cache, and $a_0$ otherwise. Notice that $\ALG_i$ incurs a cost of roughly $c^i$ in every $j$-block for any $j \geq i$. Thus, after we bound the total number of $i$-blocks (Lemma~\ref{lem:num-i-blocks}), we can conclude that $\cost(\ALG)$ is $\Omega(k)$ times the average cost of the $\ALG_i$ (Lemma~\ref{lem:det-alg-lower}). Since the optimal algorithm is no worse than the average of these $k$ algorithms, the theorem follows. We now begin with the formal analysis.

\begin{lemma} \label{lem:det-contiguous}
For every $\ell$, an $\ell$-block is a contiguous sequence of $a_0$'s, then $a_1$'s, and so on, ending with a single $a_\ell$.
\end{lemma}

\begin{proof}
It suffices to show $u_0 < u_1 < \ldots < u_k$ at Step~\ref{prp1:repeat}; this clearly holds for the initial values of the $u_i$. Thus, it suffices to prove $u_\ell < u_{\ell + 1}$ because $u_0 < u_1 < \ldots < u_\ell$ from Step~\ref{prp1:uj} and the value of $u_j$ for $j \geq \ell + 1$ remains unchanged within each step.

Suppose $u_{\ell+1} = t_0 + (2c+2)^{\ell + 1}$ for some $t_0$, and for contradiction, suppose the value of $u_\ell$ exceeds $u_{\ell+1}$ at some point $t > t_0$. Since the value of $u_{\ell+1}$ has not changed, the blocks between $t$ and $t_0$ must all be $j$-blocks for $j \leq \ell$. Furthermore, the value of $u_\ell$ only changes after we create an $\ell$-block, and each time, it increases by $(2c+2)^\ell$. However, the number of $\ell$-blocks that have appeared is at most $2c$ because of the condition in Step~\ref{prp1:inc_i}: if $y_\ell \geq 2c$, then we would have created an $(\ell+1)$-plus block. Thus, the value of $u_\ell$ is at most $t_0 + 2c \cdot (2c+2)^\ell < t_0 + (2c+2)^{\ell+1} = u_{\ell+1}$.
\end{proof}

Let $v_i$ denote the number of regular $i$-blocks, and let $v'_i$ denote the number of irregular $i$-blocks. For each $i \in \{0, 1, \ldots k\}$, $\ALG$ will miss on page $a_i$ in each of the $v_i$ regular blocks. This implies
$$\cost(\ALG)\geq v_0+v_1c+v_2c^2+\cdots+v_kc^k.$$

\begin{lemma} \label{lem:num-i-blocks}
For any $i \geq 1$, the total number of $i$-blocks is at most $\sum_{j=0}^{i}\left( \frac{v_j}{(2c)^{i-j}} \right)$.
\end{lemma}

\begin{proof}
An irregular $i$-block is created only after $2c$ $(i-1)$-blocks have been created since the last time an $i$-plus block was created. Since every $i$-block is also an $i$-plus block, the number of $(i-1)$ blocks since the last time an $i$-block was created must also be at least $2c$. So we have $v'_i\leq \frac{1}{2c}(v'_{i-1}+v_{i-1})$. (Since every $0$-block is regular, we have $v'_1\leq \frac{1}{2c}v_0$.) Adding $v_i$ to both sides and repeatedly applying this inequality proves the lemma.
\end{proof}

Now we analyze the cost of any algorithm $\ALG$ by bounding it against the performance of $k$ algorithms, defined as follows. For any $i \in \{1,2, \ldots, k\}$, the algorithm $\ALG_i$ evicts $a_i$ if it is in the cache on a cache miss, and $a_0$ otherwise.

\begin{lemma} \label{lem:det-alg-lower}
On the adversarial input generated by the procedure above, the total cost of the algorithms $\ALG_1, \ldots, \ALG_k$ is at most $32\cdot \cost(\ALG)$. That is,
\[
\sum_{i=1}^k \cost(\ALG_i) \leq 32 \cdot \cost(\ALG).
\]
\end{lemma}

\begin{proof}
Notice that $\ALG_i$ misses on request $a_i$ at most once in any $i$-plus block, so $\ALG_i$ misses on page $a_i$ at most $\sum_{j=i}^{k}(v_j+v'_j)$ times. Furthermore, $\ALG_i$ alternates between evicting $a_0$ and $a_i$, so $\ALG_i$ misses on page $a_0$ at most $1+\sum_{j=i}^{k}(v_j+v'_j)$ times in total. We bound the cost of every $\ALG_i$ miss (either on $a_0$ or $a_i$) by $(c^i + 1)$.

Thus, by Lemma~\ref{lem:num-i-blocks}, we have the following:
\begin{align*}
\cost(\ALG_i)
&\leq  \left(1+2\sum_{j=i}^{k}(v_j+v'_j)\right)\cdot(c^{i}+1) \\
&\leq  8c^{i}\sum_{j=i}^{k}(v_j+v'_j) \\
&\leq  8\sum_{j=i}^{k}\sum_{j'=0}^{j}\left( \frac{v_{j'}\cdot c^i}{(2c)^{j-j'}} \right). \\
\end{align*}

Summing across all values of $i \in \{1,2, \ldots, k\}$, we have
\begin{align*}
\sum_{i=1}^{k} \cost(\ALG_i) 
&\leq 8\sum_{i=1}^{k}\sum_{j=i}^{k}\sum_{j'=0}^{j}\left( \frac{v_{j'}\cdot c^i}{(2c)^{j-j'}} \right)  \\
&\leq 8\sum_{j'=0}^{k}\sum_{j=j'}^{k}\sum_{i=1}^{j} \left( \frac{v_{j'}\cdot c^i}{(2c)^{j-j'}} \right)  \\
&\leq 8\sum_{j'=0}^{k}v_{j'}\cdot \sum_{j=j'}^{k}\frac{1}{(2c)^{j-j'}} \cdot \sum_{i=1}^{j}c^i   \\
&\leq 16\sum_{j'=0}^{k}v_{j'}\cdot \sum_{j=j'}^{k}\frac{1}{(2c)^{j-j'}} c^j   \\
&\leq 32\sum_{j'=0}^{k}v_{j'}\cdot c^{j'}    
\leq 32\cdot \cost(\ALG).
\end{align*}
\qedhere
\end{proof}

We now conclude the proof of Theorem~\ref{thm:prp-det}. From Lemma~\ref{lem:det-alg-lower}, we have
\[
\OPT\leq \min\{\cost(\ALG_1),\cost(\ALG_2),\ldots,\cost(\ALG_k)\}\leq \frac{32}{k}\cost(\ALG),
\]
so $\ALG$ is $\Omega(k)$-competitive, as desired.

\subsection{Randomized Lower Bound} \label{sec:random_lower}
This subsection is devoted to proving the following theorem:
\begin{theorem} \label{thm:prp-rand}
For weighted paging with PRP, any randomized algorithm is $\Omega(\log k)$-competitive.
\end{theorem}

Here, we still use the same idea of request blocks, but now the input is derived from a fixed distribution and is not aware of the state of $\ALG$. The main idea is to design a distribution over block sizes in a manner that still causes any fixed deterministic algorithm $\ALG$ to suffer a large cost {\em in expectation}, and then invoke Yao's minimax principle to translate this to a randomized lower bound. Let $H_k = 1 + 1/2 + \cdots + 1/k \approx \ln k$ denote the $k$-th harmonic number. The input is defined as follows:

\begin{enumerate}
    \item For $0\leq i\leq k$, set $u_i=(2ckH_k+2)^i$ and let $y_i=0$ for $i<k$.
    \item Repeat the following:\label{prp2:repeat}
    \begin{enumerate}
        \item Select a value of $\ell$ according to the following probability distribution: $\Pr[\ell=j] = \frac{c-1}{c^{j+1}}$ for $j \in \{0,1, \ldots, k-1\}$ and $\Pr[\ell=k] = \frac{1}{c^{k}}$.\label{prp2:prob}
        \item Increase $\ell$ until $\ell = k$ or $y_\ell < 2ckH_k$. \label{prp2:inc_i}
        \item For $j$ from $0$ to $\ell$,
        \begin{enumerate}
            \item Set all requests from time $t+1$ through $u_j-1$ as $a_{j-1}$. (Note: If $j = 0$, then $u_j = t+1$, so this step is empty.)
            \item Set the request at time $u_j$ as $a_j$.
            \item Let $t = u_j$.
        \end{enumerate}
        \item For $0\leq j\leq \ell$, let $u_j=t+(2ckH_k+2)^j$.
        \item For $0\leq j<\ell$, let $y_j=0$. If $\ell<k$, increase $y_{\ell}$ by one.
    \end{enumerate}
\end{enumerate}

Note that if $\ell$ is not increased in Step \ref{prp2:inc_i}, then this block is \emph{regular}; otherwise, it is \emph{irregular}. Let $v_i$ denote the number of regular $i$-blocks, and let $v'_i$ denote the number of irregular $i$-blocks. A $j$-block is an \emph{$i$-plus} block if and only if $j \geq i$. We first lower bound the cost of $\ALG$ by the number of blocks.

\begin{lemma} \label{lem:lower-prp-alg}
Every requested block increases $\Ex{\cost(\ALG)}$ by at least a constant.
\end{lemma}

\begin{proof}
At every time step, the cache of $\ALG$ is missing some page $a_j$. The probability that $a_j$ is requested in the next block is at least $\Pr[\ell = j] \geq \frac{1}{2c^j}$, so the expected cost of serving this block is at least $c^j \cdot \Pr[\ell = j] = \Omega(1)$.
\end{proof}

For the rest of the proof, we upper bound the cost of $\OPT$. We first upper bound the number of regular blocks, and then we use this to bound the number of irregular blocks.

\begin{lemma} \label{lem:num_reg_irre}
For every $i \in \{0, 1, \ldots, k\}$, we have $\Ex{v_i}\leq 2c^{-i}m$.
\end{lemma}

\begin{proof}
Consider the potential function $\phi(y)=\sum_{i=0}^{k-1}y_i \geq 0$. The initial value of $\phi(y)$ is 0. Notice that whenever a regular block is generated, $\phi(y)$ increases by at most 1, and whenever an irregular block is generated, $\phi(y)$ decreases by at least $2ckH_k$. Thus, the number of irregular blocks is at most the number of regular blocks, so the total number of blocks is at most $2m$. The lemma follows by noting that the probability that a block is a regular $i$-block is at most $c^{-i}$.
\end{proof}

\begin{lemma} \label{lem:upper-v'_i-1}
For every $i \in \{0, 1, \ldots, k\}$, we have $\Ex{v'_i}\leq \frac{2m}{c^i k H_k}$.
\end{lemma}

\begin{proof}
Observe that $v'_i\leq \frac{1}{2ckH_k}(v'_{i-1}+v_{i-1})$ and $v'_1\leq \frac{1}{2ckH_k}v_0$. Repeatedly applying this inequality yields
\[
\Ex{v'_i}
\leq \sum_{j=0}^{i-1} \frac{\Ex{v_j}}{(2ckH_k)^{i-j}} 
\leq \sum_{j=0}^{i-1} \frac{2c^{-j}m}{(2ckH_k)^{i-j}}  
= \frac{2m}{c^i}\sum_{j=0}^{i-1} \frac{1}{(2kH_k)^{i-j}} 
\leq \frac{2m}{c^ikH_k},
\]
where the second inequality holds due to Lemma \ref{lem:num_reg_irre}.
\end{proof}

Now let $A$ denote the entire sequence of requests, $B$ the subsequence of $A$ comprising all regular blocks, and $m$ the number of blocks in $B$. We bound $\OPT = \OPT(A)$ in terms of the optimal cost on $B$ and the number of irregular blocks.

\begin{lemma}\label{lem:opta-b-c}
Let $\OPT(A)$ and $\OPT(B)$ denote the optimal offline algorithm on request sequences $A$ and $B$ respectively. Then $\cost(\OPT(A))\leq \cost(\OPT(B)) + 4c\sum_{i=0}^{k}v'_{i}c^{i}$.
\end{lemma}

\begin{proof}
Consider the following algorithm $\ALG_{A}$ on request sequence $A$:
\begin{enumerate}
    \item For requests in regular blocks, imitate $\OPT(B)$. That is, copy the cache contents when $\OPT(B)$ serves this block. \label{prp2:proof:imitate}
    \item Upon the arrival of an irregular $i$-block, let $a_\ell$ denote the page not in the cache.\label{prp2:proof:stable}
    \begin{enumerate}
        \item If $\ell>i$, then the cost of serving this block is 0.
        \item If $1\leq \ell \leq i$, evict $a_0$ when $a_\ell$ is requested. Then evict $a_\ell$ and fetch $a_0$ at the end of this block; the cost of this is $2(c^i+1)$.
        \item If $\ell=0$, we evict $a_1$ and fetch $a_0$ when $a_0$ is requested. Then we evict $a_0$ and fetch $a_1$ when $a_1$ is requested or at the end of this block (if $a_1$ is not requested in this block). The cost is $2(c+1)$.
    \end{enumerate}
\end{enumerate}

For each irregular block, notice that the cache of $\ALG_A$ is the same at the beginning and the end of the block. So Step \ref{prp2:proof:stable} does not influence the imitation in Step \ref{prp2:proof:imitate}. The cost of serving an irregular $i$-block is at most $4c^{i+1}$. Combining these facts proves the lemma.
\end{proof}

To bound $\OPT(B)$, we divide the sequence $B$ into phases. Each phase is a contiguous sequence of blocks. Phases are defined recursively, starting with $0$-phases all the way through to $k$-phases. A 0-phase is defined as a single request. For $i\geq 1$, let $M_i$ denote the first time that an $i$-plus-block is requested and let $Q_i$ denote the first time that $c$ $(i-1)$-phases have appeared. An $i$-phase ends immediately after $M_i$ and $Q_i$ have both occurred. In other words, an $i$-phase is a minimal contiguous subsequence that contains $c$ $(i-1)$-phases and an $i$-plus block. (Notice that for a fixed $i$, the set of $i$-phases partition the input sequence.)

For any $k$-phase, we upper bound $\OPT$ by considering an algorithm $\ALG^k_B$ that is optimal for $B$ subject to the additional restriction that $a_0$ is not in the cache at the beginning or end of any $k$-phase. We bound the cost of $\ALG^k_B$ in any $k$-phase using a more general lemma.

\begin{lemma} \label{lem:opt-B}
For any $i$, let $\ALG^i_B$ be an optimal algorithm on $B$ subject to the following: $a_0$ is not in the cache at the beginning or the end of any $i$-phase. Then the cost of $\ALG^i_B$ within an $i$-phase is at most $4c^{i+1}$. In particular, in each $k$-phase, the algorithm $\ALG^k_B$ incurs cost at most $4c^{k+1}$.
\end{lemma}

\begin{proof}
We shall prove this by induction on $i$. If $i = 0$, then the phase under consideration is one step. To serve one step, we can evict $a_1$ to serve $a_0$, and then evict $a_0$ if necessary for a total cost of $4c$. Now assume that the lemma holds for all values in $\{0, \ldots, i-1\}$. Let $s_i$ denote the first $i$-plus block; there are two possible cases for the structure of an $i$-phase:
\begin{enumerate}
    \item $s_i$ appears after the $c$ $(i-1)$-phases: In this case, the $i$-phase ends after this block. Thus, one strategy to serve the phase is to evict $a_i$ at the beginning and evict $a_0$ when $a_i$ is requested within $s_i$. These two evictions cost at most $4c^{i+1}$.
    \item $s_i$ appears within the first $c$ $(i-1)$-phases: By the inductive hypothesis, the algorithm can serve these $c$ $(i-1)$-phases with total cost at most $c \cdot 4c^{i} = 4c^{i+1}$.\qedhere
\end{enumerate}
\end{proof}

Finally, we lower bound the expected number of blocks in an $i$-phase. Since the total number of blocks is fixed, this allows us to upper bound the number of $k$-phases in the entire sequence. The next proposition forms the technical core of the lower bound:

\begin{proposition} \label{prop:L_k-lower}
For $i \geq 1$, the expected number of blocks in an $i$-phase is at least $c^iH_i / 4$.
\end{proposition}

We defer the proof of Proposition~\ref{prop:L_k-lower} to the end of this section; first, we use it to prove Theorem~\ref{thm:prp-rand}.

\begin{proof}[Proof of Theorem~\ref{thm:prp-rand}]
Let $\OPT(A)$ denote the cost of an optimal algorithm on the request sequence $A$, and let $\OPT(B)$ denote the cost of an optimal algorithm on the regular blocks $B$. Then we have the following:
\begin{align*}
\Ex{\cost(\OPT(A))}
&\leq \Ex{\cost(\OPT(B))} + 4c\sum_{i=0}^{k}c^{i}\cdot \Ex{v'_{i}} \tag{Lemma~\ref{lem:opta-b-c}} \\
&\leq \Ex{\cost(\ALG^k_B)} + 4c\sum_{i=0}^{k}c^{i}\cdot \frac{2m}{c^ikH_k} \tag{Lemma~\ref{lem:upper-v'_i-1}} \\
&\leq 4c^{k+1}\cdot \Ex{N_k(B)} + \frac{16cm}{H_k}, \tag{Lemma~\ref{lem:opt-B}}
\end{align*}
where $N_k(B)$ denotes the number of $k$-phases in $B$. According to Proposition~\ref{prop:L_k-lower}, the expected number of blocks in a $k$-phase is at least $c^kH_k/4$, which implies $\Ex{N_k(B)} \leq \frac{ 4m}{c^kH_k}$. Combining this with the above, we get
\[
\Ex{\cost(\OPT(A))} \leq \frac{16cm}{H_k} + \frac{16cm}{H_k} = O\left(\frac{m}{H_k}\right).
\]
Since any algorithm incurs at least some constant cost in every block by Lemma~\ref{lem:lower-prp-alg}, its cost is $\Omega(m)$, which concludes the proof.
\end{proof}

\paragraph*{Proof of Proposition~\ref{prop:L_k-lower}} \label{sec:L_k-lower-proof}
Let $z_i$ be a random variable denoting the number of $i$-plus blocks in a fixed $i$-phase. We will first prove a sequence of three lemmas to yield a lower bound on $\Ex{z_i}$.

\begin{lemma}\label{prp:z=z+p}
For any $i \geq 1$, we have $\Ex{z_i}=\Ex{z_{i-1}}+\Pr{M_i>Q_i}$.
\end{lemma}

\begin{proof}
Recall that an $i$-phase ends once it contains $c$ $(i-1)$-phases and an $i$-plus block. In each of the $(i-1)$-phases, the expected number of $(i-1)$-plus blocks is $\Ex{z_{i-1}}$, so the total expected number of $(i-1)$-plus blocks in the first $c$ $(i-1)$-phases of an $i$-phase is $c\cdot \Ex{z_{i-1}}$.

An elementary calculation shows that an $(i-1)$-plus block is an $i$-plus block with probability $1/c$. Thus, in expectation, the first $c$ $(i-1)$-phases of this $i$-phase contain $\Ex{z_{i-1}}$ $i$-plus blocks.

If there are no $i$-plus blocks in the first $c$ $(i-1)$-phases, then the $i$-phase ends as soon as an $i$-plus block appears. In this case, we have $z_i = 1$, and this happens with probability exactly $\Pr{M_i>Q_i}$. Otherwise, the $i$-phase ends immediately after the $c$ $(i-1)$-phases, in which case no additional term is added.
\end{proof}

\begin{lemma} \label{lem:p-mi-qi}
For any $i \geq 1$, we have $\Pr{M_i>Q_i}\geq e^{-2\Ex{z_{i-1}}}$.
\end{lemma}

\begin{proof}
We let $v_1, \ldots, v_c$ denote the number of $i$-plus blocks in the first $c$ $(i-1)$-phases and let $V = \sum_{i=1}^c v_i$. As we saw in the proof of Lemma \ref{prp:z=z+p}, an $(i-1)$-plus block is an $i$-plus block with probability $1/c$, so the probability that an $(i-1)$-plus block is an $(i-1)$-block is $1 - 1/c$. Thus, we have
\[
\Pr{M_i>Q_i} =\mathbb{E}_{v_1,v_2,\dots,v_c}\left[\left(1-\frac{1}{c}\right)^V\right] \geq \left(1-\frac{1}{c}\right)^{\Ex{V}} = \left(1-\frac{1}{c}\right)^{c\cdot \Ex{z_{i-1}}}
\]
where the inequality follows from convexity and the second equality holds due to linearity of expectation. The lemma follows from this and the fact that $c \geq 2$.
\end{proof}

\begin{lemma}\label{lem:zi-hi}
For any $i \geq 0$, we have $\Ex{z_i}\geq \frac{1}{4}H_i$.
\end{lemma}

\begin{proof}
When $i \leq 4$, we have $\Ex{z_i}\geq 1\geq \frac{1}{4}H_i$. Now for induction, assume the statement holds for $j < i$, and consider the two possible cases:
\begin{enumerate}
\item If $\Ex{z_{i-1}}\geq \frac{1}{2}H_{i-1}$, then Lemma \ref{prp:z=z+p} implies $\Ex{z_i}\geq \Ex{z_{i-1}}\geq \frac{1}{4}H_i$.
\item If $\Ex{z_{i-1}}< \frac{1}{2}H_{i-1} < \frac{1}{2}(1 + \ln(i-1))$, then\\
$\Ex{z_i} =\Ex{z_{i-1}}+\Pr{M_i>Q_i} \geq \frac{1}{4}H_{i-1} + e^{-2\cdot \Ex{z_{i-1}}}$,
where the equality follows from Lemma~\ref{prp:z=z+p} and the inequality holds by the induction hypothesis and Lemma~\ref{lem:p-mi-qi}. Thus,\\ 
$\Ex{z_i} \geq \frac{1}{4}H_{i-1} + \frac{1}{e}\cdot\frac{1}{i-1}\geq \frac{1}{4}H_{i}$.\qedhere
\end{enumerate}
\end{proof}

Now let $L_i$ denote the number of blocks in an $i$-phase; recall that our goal is to lower bound its expectation by $c^iH_i/4$. The following lemma relates $L_i$ to $z_i$.

\begin{lemma}\label{lem:L=cz}
For any $i\geq 0$, we have $\Ex{L_i}=c^i\cdot \Ex{z_i}$.
\end{lemma}

\begin{proof}
When $i=0$, the lemma holds because $E[L_0]=E[z_0]=1$, so now we assume $i\geq 1$. Recall that an $i$-phase contains at least $c$ $(i-1)$-phases, so the expected total number of blocks in the first $c$ $(i-1)$-phases of this $i$-phase is $c\cdot \Ex{L_{i-1}}$.

If there are no $i$-plus-blocks in these $c$ $(i-1)$-phases, we need to wait for an $i$-plus block to appear in order for the $i$-phase to end. This is a geometric random variable with expectation $c^i$. Thus, we have:
$\Ex{L_i}=c\cdot \Ex{L_{i-1}}+c^i \cdot \Pr{M_i>Q_i}$.
Applying this recursively,
\[
\Ex{L_i}=c^i\left(\sum_{j=1}^{i}\Pr{M_j>Q_j}+\Ex{L_0}\right) = c^i\left(\sum_{j=1}^{i}\Pr{M_j>Q_j}+1\right)
\]
Furthermore, from Lemma \ref{prp:z=z+p}, we have
\[
\Ex{z_i} = \Ex{z_{i-1}} + \Pr{M_i > Q_i} = \Ex{z_0} + \sum_{j=1}^i \Pr{M_j > Q_j} = 1 + \sum_{j=1}^i \Pr{M_j > Q_j}.
\]
Combining the two equalities yields the lemma.
\end{proof}
We conclude by proving Proposition~\ref{prop:L_k-lower}. Fix some $i \geq 1$. Using 
Lemma~\ref{lem:L=cz} and Lemma~\ref{lem:zi-hi}, we get 
$\Ex{L_i} = c^i\cdot \Ex{z_i} \geq \frac{c^iH_i}{4}$.

\section{The $\ell$-Strong Lookahead Model} \label{sec:l-strong}
Now we consider the following prediction model: at each time $t$, the algorithm can see request $p_t$ as well as $L(t)$, which is the set of all requests through the $\ell$-th distinct request. In other words, the algorithm can always see the next contiguous subsequence of $\ell$ distinct pages (excluding $p_t$) for a fixed value of $\ell$. This model was introduced by Albers~\cite{albers1993influence}, who (among other things) proved the following lower bounds on algorithms with $\ell$-strong lookahead.

\begin{lemma}[\cite{albers1993influence}]\label{lem:strong-1}
For unweighted paging with $\ell$-strong lookahead where $\ell \leq k-2$, any deterministic algorithm is $\Omega(k-\ell)$-competitive. For randomized algorithms, the bound is $\Omega(\log(k-\ell))$.
\end{lemma}

Notice that Lemma~\ref{lem:strong-1} implies that for small values of $\ell$, $\ell$-strong lookahead provides no asymptotic improvement to the competitive ratio of any algorithm. The proof proceeds by constructing a particular sequence of requests and analyzing the performance of any algorithm on this sequence. By slightly modifying the sequence, we can prove a similar result for the weighted paging problem.

\begin{theorem} \label{thm:ell-large}
For weighted paging with $\ell$-strong lookahead where $n-k+1\leq \ell\leq n-1$, any deterministic algorithm is $\Omega(n-\ell)$-competitive, and any randomized algorithm is $\Omega(\log (n-\ell))$-competitive.
\end{theorem}

\begin{proof}
We modify the adversarial input in Lemma~\ref{lem:strong-1} as follows: insert $n-k-1$ distinct pages with very low weight between every two pages. This causes the lookahead to have effective size $\ell' = \ell-(n-k-1)$, because at any point $L(t)$ contains at most $\ell'$ pages with normal weight. Note that if $\ell \leq n-k$, then $\ell' \leq 1$, and from Lemma~\ref{lem:strong-1}, a lookahead of size 1 provides no asymptotic benefit to any algorithm.

If $\ell \leq n-3$, then $\ell' \leq k-2$. Thus, we can apply Lemma~\ref{lem:strong-1} to conclude that for any deterministic algorithm, the competitive ratio is $\Omega(k - \ell') = \Omega(n-\ell-1)$, and for any randomized algorithm, the competitive ratio is $\Omega(\log(n-\ell-1))$. Otherwise, if $\ell \geq n-2$, then the lower bounds continue to hold because when $\ell = n-3$, they are $\Omega(1)$.
\end{proof}

\section{The Strong Per-Request Prediction Model (SPRP)} \label{sec:static}
In this section, we define a simple algorithm called \static that is 2-competitive when the SPRP predictions are always correct. At any time step $t$, let $L(t)$ denote the set of pages in the current prediction. The \static algorithm runs on ``batches'' of requests. The first batch starts at $t=1$ and comprises all requests in $L(1)$. The next batch starts once the first batch ends, i.e. at $|L(1)|+1$, and comprises all predicted requests at that time, and so on. Within each batch, the \static algorithm runs the optimal offline strategy, computed at the beginning of the batch on the entire set of requests in the batch.  

\begin{theorem} \label{thm:static}
The \static algorithm is 2-competitive when the predictions from SPRP are entirely correct.
\end{theorem}

\begin{proof}
In this proof, we assume w.l.o.g. that evicting page $p$ costs $w(p)$, and fetches can be performed for free.
%

Suppose the algorithm runs a total of $m$ batches $B_1, \ldots, B_m$. Consider a page $p$ in some batch $B_i$ where $i < m$. If $p$ appears again after $B_i$, then upon seeing the last request for $p$ in $B_i$, SPRP will include $p$ in the next batch $B_{i+1}$. (If $p$ does not appear again, then the next batch must be the last batch.) Therefore, the batches satisfy $B_1 \subseteq B_2 \subseteq \cdots \subseteq B_{m-1}$.

Now let $\OPT$ denote a fixed optimal offline algorithm for the entire sequence, and let $\OPT_i$ denote the cost of $\OPT$ incurred in $B_i$. Similarly, let $S$ denote the total cost of \static, and let $S_i$ denote the cost that \static incurs in $B_i$. So we have $\OPT = \sum_{i=1}^m \OPT_i$ and $S = \sum_{i=1}^m S_i$.

Fix a batch index $j \in \{2, 3, \ldots, m\}$ and let $C(\OPT_{j-1})$ and $C(S_{j-1})$ denote the cache states of $\OPT$ and \static immediately before batch $B_j$. We know that \static runs an optimal offline algorithm on $B_j$. One feasible solution is to immediately change the cache state to $C(\OPT_{j-1})$, and then imitate what $\OPT$ does to serve $B_j$. Since we charge for evictions, we have
\[
S_j \leq \OPT_j + \sum_{p \in C(S_{j-1}) \setminus C(\OPT_{j-1})} w(p), \text{ for every } j \in \{2, 3, \ldots, m\}.
\]
Consider some $p \in C(S_{j-1})\setminus C(\OPT_{j-1})$: since $p \in C(S_{j-1})$, we know $p$ must have appeared before the start of $B_j$ (because \static does not fetch pages that have never been requested). Since $B_{j-1}$ contains all pages that appeared before, in particular, $p$ must be in $B_{j-1}$. Furthermore, since $p \not\in C(\OPT_{j-1})$, then at some point while serving $B_{j-1}$, $\OPT$ must have evicted $p$. Thus, $S_j \leq \OPT_j + \OPT_{j-1}$. Summing over all $j \geq 2$ and $S_1 \leq \OPT_1$ proves the theorem.
\end{proof}

\section{The SPRP Model with Prediction Errors} \label{sec:sprp}
In this section, we consider the SPRP prediction model with the possibility of prediction errors. We first define three measurements of error and then prove lower and upper bounds on algorithms with imperfect SPRP, in terms of these error measurements.

Let $A$ denote a prediction sequence of length $m$, and let $B$ denote an input sequence of length $n$. For any time $t$, let $A_t$ and $B_t$ denote the $t$-th element of $A$ and $B$, respectively. We also define the following for any time step $t$:
\begin{itemize}
    \item $\prev(t)$: The largest $i < t$ such that $B_i = B_t$ (or 0 if no such if no such $i$ exists).
    \item $\next(t)$: The smallest $i > t$ such that $B_i = B_t$ (or $n+1$ if no such $i$ exists).
    \item $\predictnext(t)$: The smallest $i > t$ such that $A_i = B_t$ (or $m+1$ if no such $i$ exists).
    \item We say two requests $A_i = B_j = p$ can be \emph{matched} only if $\predictnext(\prev(j)) = i$. Furthermore, no edges in a matching are allowed to cross. In other words, $A_i$ must be the earliest occurrence of $p$ in $A$ after the time of the last $p$ in $B$ before $B_j$.
\end{itemize}
First, we define a variant of edit distance between the two sequences.
\begin{definition}
The \emph{edit distance} $\ell_{ed}$ between $A$ and $B$ is the total minimum weight of unmatched elements of $A$ and $B$.
\end{definition}
Next, we define an error measure based on the metric $1$-norm distance between corresponding requests on the standard weighted star metric denoting the weighted paging problem.
\begin{definition}
The \emph{1-norm distance} $\ell_{1}$ between $A$ and $B$ is defined as follows:
\begin{align*}
\ell_1=\sum\limits_{\substack{i=1 \\ A_i \neq B_i}}^n \left(w(A_i)+w(B_i)\right). \tag{1-norm}
\end{align*}
\end{definition}
Third, we define an error measure inspired by the PRP model that was also used in \cite{lykouris2018competitive}.
\begin{definition}
The \emph{prediction distance} $\ell_{pd}$ between $A$ and $B$ is defined as follows:
\begin{align*}
\ell_{pd}=\sum_{i=1}^{n} w(B_i)\cdot \left|\next(i)-\predictnext(i)\right|.
\end{align*}
\end{definition}

\subsection{Lower Bounds} \label{subsec:sprp-lower}
In this section, we give an overview of the lower bounds stated in Theorems~\ref{thm:1-lpd-lower},~\ref{thm:1-l1-lower}, and~\ref{thm:1-led-lower}. We focus on the $\ell_{ed}$ (i.e., Theorem~\ref{thm:1-led-lower}) error measurement; the proofs for $\ell_1$ and $\ell_{pd}$ follow similarly.

Our high-level argument proceeds as follows: recall that in Section \ref{sec:PRP}, we showed a lower bound of $\Omega(k)$ on the competitive ratio of deterministic PRP-based algorithms. Given an SPRP algorithm $\ALG$, we design a PRP algorithm $\ALG'$ specifically for the input generated by the procedure described in Section \ref{sec:PRP}. (Recall that this input is a sequence of blocks, where a \emph{block} is a string of $a_0$'s, $a_1$'s, and so on, ending with a single page $a_\ell$ for some $\ell$.)

We show that if $\ALG$ has cost $o(k) \cdot \OPT + o(\ell_{ed})$ (where $\OPT$ is the optimal cost of the SPRP instance), then $\ALG'$ will have cost $o(k)\cdot \OPT'$ (where $\OPT'$ is the optimal cost of the PRP instance), which contradicts our PRP lower bound of $\Omega(k)$ on this input. For the randomized lower bound, we use the same line of reasoning, but replace $\Omega(k)$ with $\Omega(\log k)$.

Let $k'$ denote the cache size of $\ALG'$. Recall that the set of possible page requests received by $\ALG'$ is $A = \{a_0, a_1, \ldots, a_{k'}\}$ where $w(a_i) = c^i$ for some constant $c \geq 2$. The oracle $\ALG$, maintained by $\ALG'$, has cache size $k = k' + 1$. The set of possible requests received by $\ALG$ is $A \cup \{b\}$ where $w(b) = 1/v$ for some sufficiently large value of $v$. (Thus, the instance for $\ALG$ has $k+1$ distinct pages.) Our PRP algorithm $\ALG'$ must define a prediction and an input sequence for $\ALG$.

\paragraph{The prediction sequence for $\ALG$:}
For any strings $X$ and $Y$, let $X + Y$ denote the concatenation of $X$ and $Y$ and let $\lambda \cdot X$ denote the concatenation of $\lambda$ copies of $X$. Let $L = 2ck'H_{k'} + 1$, and consider the series of strings: $S_0 = 2 \cdot a_0$, and $S_i = L \cdot S_{i-1} + a_i$ for $i \in \{1, \ldots, k'\}$. We fix $S \coloneqq M \cdot S_{k'}$, for some sufficiently large $M$, as the prediction sequence for the SPRP algorithm. (Observe that $S$ only contains $k$ distinct pages, and the oracle $\ALG$ has cache size $k$.)

\paragraph{$\ALG'$ and the request sequence for $\ALG$:}
Our PRP algorithm $\ALG'$ will simultaneously construct input for $\ALG$ while serving its own requests. Since randomized and fractional algorithms are equivalent up to constants (see Bansal et al.~\cite{bansal2012primal}), we view the SPRP algorithm $\ALG$ from a fractional perspective. Let $q_i \in [0,1]$ denote the fraction of page $a_i$ not in the cache of $\ALG$. Notice that the vector $q=(q_0,q_1,\dots,q_{k'})$ satisfies $\sum_{i=0}^{k'}q_i\geq 1$. (A deterministic algorithm is the special case where every $q_i\in \{0,1\}$.) Similarly, let $q'=(q'_0,q'_1,\dots,q'_{k'})$, where $q'_i$ denotes the amount of request for $a_i$ that is not in the cache in $\ALG'$. 

When a block ending with $a_i$ is requested, $\ALG'$ scans $S$ for the next appearance of $a_i$. It then feeds the scanned portion to $\ALG$, followed by a single request for page $b$. In this case, the prediction error only occurs due to the requests for this page $b$. After serving this request $b$, the cache of $\ALG$ contains at most $k'$ pages in $A$. This enables $\ALG'$ to mimic the behavior of $\ALG$ upon serving the current block. This process continues for every block: $\ALG'$ modifies the input by inserting an extra request $b$ into the input for $\ALG$, and mimics the resulting cache state of $\ALG$. The details of our algorithm $\ALG'$ are given below:

\begin{enumerate}
    \item Initially, let $S$ be the input for $\ALG$ and $t = 0$. (We will modify $S$ as time passes.)
    \item For all $0\leq i\leq k'$, let $q'_i=1$. (Note that the initial value of every $q_i$ is also 1.)
    \item On PRP request block $s_i = (a_0, a_1, \ldots, a_i)$ (for some unknown $i$): \label{reduction_suffer_cost}
    \begin{enumerate}
        \item Let $q'=(q'_0,q'_1,\dots,q'_{k'})$ denote the current cache state. \label{reduction_c1}
        \item Set $q'=(0,\min\{1,q'_0+q'_1\},q'_2,q'_3,\dots,q'_{k'})$ to serve $a_0$. Note that after we serve $a_0$, the PRP prediction tells us the value of $i$. \label{reduction_c2}
        \item Find the first time $t'$ after $t$ when $S$ requests $a_{i}$ and set $t=t'+2$.
        \item Change the request at time $t$ into $b$. (Note that the original request is $a_0$.)\label{reduction_insert}
        \item Run $\ALG$ until this $b$ is served to obtain a vector $q=(q_0,q_1,\dots,q_{k'})$.\label{reduction_run_black}
        \item If $i\geq 1$, set $q'=(\min\{1,\sum_{j=0}^{i}q'_j\},0,0,\dots,0,q'_{i+1},q'_{i+2},\dots,q'_{k'})$; this serves the requests $(a_1,a_2,\dots,a_i)$.\label{reduction_c3}
        \item Set $q'=(q_0,q_1,\dots,q_{k'})$.\label{reduction_c4}
    \end{enumerate}
\end{enumerate}

\smallskip\noindent
{\bf Bounding the costs.} 
The main idea in the analysis is the following: since the input sequences to $\ALG$ and $\ALG'$ are closely related, and they maintain similar cache states, we can show that they are coupled both in terms of the algorithm's cost and the optimal cost. Therefore, the ratio of $\Omega(k)$ for $\ALG'$ (from Theorem~\ref{thm:prp-det}) translates to a ratio of $\Omega(k)$ for $\ALG$. Furthermore, since the only prediction errors are due to the additional requests for page $b$, and this page has a very small weight, the cost of $\ALG$ is at least the value of $\ell_{ed}$.
(The same line of reasoning is used for randomized algorithms, but $\Omega(k)$ is replaced by $\Omega(\log k)$.)

We now formalize the above line of reasoning with the following lemmas.

\begin{lemma} \label{lem:reduction-alg}
Using any SPRP algorithm $\ALG$ as a black box, the PRP algorithm $\ALG'$ satisfies the following: $\cost(\ALG')\leq 2(c+1)\cdot \cost(\ALG)$.
\end{lemma}

\begin{proof}
Note that $q=q'$ at the beginning and end of Step \ref{reduction_suffer_cost}. For convenience, let $q'$ denote the vector at the beginning of Step \ref{reduction_suffer_cost}, and let $q$ denote the vector at the end of Step \ref{reduction_suffer_cost}. Let $\cost_{\ALG}$ and $\cost_{\ALG'}$ denote the cost of $\ALG$ and $\ALG'$ respectively incurred in a fixed Step \ref{reduction_suffer_cost}.

Each time $\ALG'$ enters Step \ref{reduction_suffer_cost}, the cost incurred is at most:
\begin{align*}
&\text{Step \ref{reduction_c2}: } q'_0\cdot (1+c),   \\
&\text{Step \ref{reduction_c3}: } (q'_0+q'_1)\cdot (1+c) + \sum_{j=2}^{i} q'_j\cdot (1+c^j),   \\
&\text{Step \ref{reduction_c4}: } \left(\sum_{j=1}^{i} q_j\cdot (1+c^j)\right)
   + \left(\sum_{j=i+1}^{k} \left|q'_j-q_j\right|\cdot (1+c^j)\right).   \\
\end{align*}
Summing the above yields the following:
\[
\cost_{\ALG'} \leq 2(c+1)\cdot \left( \left( \sum_{j=0}^{i} c^j\cdot \left( q_j+q'_j \right) \right) + \left( \sum_{j=i+1}^{k} c^j\cdot \left| q_j-q'_j \right|  \right)   \right).
\]

Now we consider $\ALG$. For each $j$, at the beginning of Step \ref{reduction_suffer_cost}, there is $q'_j$ amount of $a_j$ not in the cache, and at the end of Step \ref{reduction_suffer_cost}, there is $q_j$ amount of $a_j$ not in the cache.

If $j>i$, the cost incurred due to $a_j$ is at least $c^j\cdot\abs{q_j-q'_j}$. If $j\leq i$, $\ALG'$ must serve $a_j$ at some point in Step \ref{reduction_run_black}, so the incurred cost due to $a_j$ is at least $c^j\cdot (q_j+q'_j)$. Summing the above yields the following:
\[
\cost_{\ALG} \geq \left( \sum_{j=0}^{i} c^j\cdot \left( q_j+q'_j \right) \right) + \left(\sum_{j=i+1}^{k} c^j\cdot \left| q_j-q'_j \right|  \right).
\]
Combining the two inequalities above proves the lemma.
\end{proof}

Now let $\OPT$ denote the optimal SPRP algorithm for the input sequence served by $\ALG$, and let $\OPT'$ denote the optimal PRP algorithm for the input sequence served by $\ALG'$. We can similarly prove the following lemma that bounds the costs of $\OPT$ and $\OPT'$ against each other.

\begin{lemma} \label{lem:reduction-opt}
The algorithms $\OPT$ and $\OPT'$ satisfy $\cost(\OPT) \leq 2\cdot \cost(\OPT')$.
\end{lemma}

\begin{proof}
Using $\OPT'$ as an oracle, we can design a potential algorithm for $\OPT$:
\begin{enumerate}
    \item Let $S$ be the initial input sequence for $\ALG$ and let $t = 0$.
    \item For all $0\leq i\leq k'$, let $q_i=1$. Note that $q'_i=1$ at the beginning.
    \item For each PRP block $s_i = (a_0, a_1, \ldots, a_i)$:
    \begin{enumerate}
        \item Find the first time $t'$ after $t$ when $S$ requests $a_{i}$. Let $t=t'+2$; note that $S_{t} = b$.
        \item Run $\OPT'$ to serve request $(a_0,a_1,\dots,a_i)$ and obtain $q'=(q'_0,q'_1,\dots,q'_{k'})$. \label{reduction_suffer_cost_opt}
        \begin{enumerate}
            \item Let $q=(q_0,q_1,\dots,q_{k'})$ denote the current cache state (i.e., immediately before we serve $a_0$). \label{reduction_opt_c1}
            \item Set $q=(0,0,\dots,0,q'_{i+1},q'_{i+2},\dots,q'_{k'})$ to serve all requests until the requested $b$.\label{reduction_opt_c2}
            \item Set $q=(q'_0,q'_1,\dots,q'_{k'})$ to serve the $b$.\label{reduction_opt_c3}
        \end{enumerate}
    \end{enumerate}
\end{enumerate}

Note that we have $q=q'$ at the beginning and the end of Step \ref{reduction_suffer_cost} in $\ALG$. For convenience, let $q'$ denote the vector at the beginning of Step \ref{reduction_suffer_cost}, and let $q$ to denote the vector at the end of Step \ref{reduction_suffer_cost}. Furthermore, let $\cost_{\OPT}$ and $\cost_{\OPT'}$ denote the cost that $\OPT$ and $\OPT'$ respectively incur in a fixed Step \ref{reduction_suffer_cost_opt}.

Each time $\OPT$ enters Step \ref{reduction_suffer_cost}, the incurred cost is at most:
\begin{align*}
&\text{Step \ref{reduction_opt_c2}: } \sum_{j=0}^{i} q'_j\cdot \left(\frac{1}{v}+c^j\right),   \\
&\text{Step \ref{reduction_opt_c3}: } \sum_{j=0}^{i} q_j\cdot \left(\frac{1}{v}+c^j\right) + \sum_{j=i+1}^{k} (\frac{1}{v}+c^j)\cdot \left| q_j-q'_j \right|.
\end{align*}

Summing the above yields the following:
\[\cost_{\OPT}\leq 2\left(\left( \sum_{j=0}^{i} c^j\cdot \left( q_j+q'_j \right) \right) + \left( \sum_{j=i+1}^{k} c^j\cdot \left| q_j-q'_j \right|  \right)\right).  \]

Now we consider $\OPT'$. At the beginning of Step \ref{reduction_suffer_cost_opt}, there is $q'_j$ amount of $a_j$ is not in the cache, and at the end of Step \ref{reduction_suffer_cost_opt}, there is $q_j$ amount of $a_j$ is not in the cache.

If $j>i$, the cost incurred due to $a_j$ is at least $c^j\cdot \abs{q_j-q'_j}$. If $j\leq i$, $\OPT'$ must serve $a_j$ while it serving $(a_0,a_1,\dots,a_i)$, so the cost due to $a_j$ is at least $c^j\cdot (q_j+q'_j)$. Summing the above yields the following:
\[
\cost_{\OPT'} \geq  \left( \sum_{j=0}^{i} c^j\cdot \left( q_j+q'_j \right) \right) + \left(\sum_{j=i+1}^{k} c^j\cdot \left| q_j-q'_j \right|\right).
\]
Combining the above inequalities proves the lemma.
\end{proof}

We are now ready to bound the cost of any algorithm with SPRP.

\begin{theorem}\label{thm:lower_bound_ed}
For weighted paging with SPRP, there is no deterministic algorithm whose cost is $o(k) \cdot \OPT + o(\ell_{ed})$, and there is no randomized algorithm whose cost is $o(\log k)\cdot \OPT + o(\ell_{ed})$.
\end{theorem}

\begin{proof}
From Theorem~\ref{thm:prp-det}, we know $\ALG'= \Omega(k) \cdot \OPT'$, so we can apply Lemmas~\ref{lem:reduction-alg} and~\ref{lem:reduction-opt} to conclude $\ALG= \Omega(k)\cdot \OPT$. Furthermore (as we saw in Section \ref{sec:PRP}), each PRP block increases $\ALG$ by at least a constant. At the same time, for each block, $\ell_1$ increases by at most $2$, because only one request is changed from $a_0$ to $b$. As a result, we can conclude $\ALG = \Omega(\ell_1)$. Similarly, for $\ell_{pd}$, notice that the only mispredictions are due to $a_0$ and $b$. This allows us to conclude $\ell_{pd} = \Theta(\ell_1)$. Finally, we can also see that in this instance, we have $\ell_{ed} = \ell_1$, so the bound continues to hold. For randomized algorithms, the same line of reasoning holds with $\Omega(\log k)$ instead of $\Omega(k)$.
\end{proof}

\subsection{Upper Bounds} \label{subsec:sprp-upper}
In this section, we give algorithms whose performance degrades with the value of the SPRP error. In particular, we first prove the upper bound in Theorem~\ref{thm:1-led-upper} for the $\ell_{ed}$ measurement, and then analyze the \follow algorithm, which proves the upper bound in Theorem~\ref{thm:1-l1-upper}.

Now we present an algorithm that uses a cache of size $k+1$ whose cost scales linearly with $\OPT + \ell_{ed}$. Following our previous terminology, let $A$ denote a prediction sequence of length $m$, and let $B$ denote an input sequence of length $n$.

Our algorithm, which we call \textsc{Learn}, relies on an algorithm that we call \textsc{Idle}. At a high level, \textsc{Idle} resembles \textsc{Static} (see Section~\ref{sec:static}): it partitions the prediction sequence $A$ into batches and runs an optimal offline algorithm on each batch. The \learn algorithm tracks the cost of imitating \textsc{Idle}: if the cost is sufficiently low, then it will imitate \textsc{Idle} on $k$ of its cache slots; otherwise, it will simply evict the page in the extra cache slot.

Before formally defining \textsc{Idle}, we consider a modified version of caching. Our cache has $k + 1$ slots, where one slot is \emph{memoryless}: it always immediately evicts the page it just fetched. In other words, this slot can serve any request, but it cannot store any pages. Let $\OPT^{+1}$ denote the optimal algorithm that uses a memoryless cache slot.

\begin{lemma}\label{kplus:corelemma1}
For any sequences $A$ and $B$, $\cost(\OPT^{+1}(A))\leq \cost(\OPT(B)) + 2\ell_{ed}$,
where $\ell_{ed}$ is the edit distance between $A$ and $B$.
\end{lemma}

\begin{proof}
Let $M$ denote the optimal matching between $A$ and $B$ (for $\ell_{ed}$). One algorithm for $\OPT^{+1}(A)$ is the following: imitate what $\OPT(B)$ does for requests matched by $M$, and use the memoryless slot for unmatched requests. The cost of this algorithm is $\OPT(B)+2\ell_{ed}$.
\end{proof}

Recall that the \static algorithm requires the use of an optimal offline algorithm. Similarly, for our new problem with a memoryless cache slot, we require a constant-approximation offline algorithm on $A$. This can be obtained from the following lemma:

\begin{lemma} \label{lem:const-plus1-pred}
Given a prediction sequence $A$, there is a randomized offline algorithm whose cost is at most a constant times the cost of $\OPT^{+1}(A)$.
\end{lemma}

\begin{proof}
Let $x(i,j)$ be an indicator variable that is 1 if page $i$ is evicted between the $j$-th time and the $(j+1)$-th time it is requested, and 0 otherwise. For any time $t\leq T$, let $B(t)=\{i|r(i,t)\geq 1\}$, where $r(i,t)$ denotes the number of times page $i$ is requested until time $t$. The problem has the following linear programming formulation:
\begin{gather*}
    \min \sum_{i=1}^{n}\sum_{j=1}^{r(i,T)}w(i)x(i,j)    \\
    \text{For any time $t$: }\sum_{i\in B(t)} x(i,r(i,t))\geq |B(t)|-k  \\
    \text{For any $i,j$: } 0\leq x(i,j)\leq 1
\end{gather*}
Recall that the size of the cache is $k+1$, including a memoryless cache slot. The constraint specifies that at any time $t$, at least $|B(t)|-k$ pages are not the normal cache, which means at most $k$ pages are in the normal cache. If the requested page is not in the normal cache but contributes to the sum in the constraint, then this corresponds to fetching it into the memoryless cache slot.

This formulation gives us a fractional solution, and for the standard caching problem (with $k$ slots), Bansal et al.~\cite{bansal2012primal} showed how to convert a fractional solution to a randomized solution while losing only a constant factor. Thus, this formulation yields a randomized integral solution for the $k$ normal slots. Note that if a requested page is not fetched by one of the $k$ normal slots, then we fetch it using the memoryless slot.

Now we analyze the cost of our algorithm in two parts: the total cost incurred by the normal slots $w_1$, and the total cost incurred by the memoryless slot $w_2$. We let $w^f_1$ and $w^f_2$ denote the corresponding costs of the fractional solution. Note that $w_1 = O(w^f_1)$ due to the rounding scheme of Bansal et al.~\cite{bansal2012primal}. Now we consider $w_2$. On the arrival of a page $p$, $w_2$ increases by $w(p)$ if it is not in the normal slots, and suppose this occurs with some probability $q$. By the rounding scheme, this means $p$ is in a normal slot with probability $1-q$, so in the fractional solution, a $1-q$ fraction of $p$ is in the normal slots. Therefore, a $q$ fraction of page $p$ is not in the normal slots, so $w^f_2$ increases by $q\cdot w(p)$, so this upper bounds the expected increase of $w_2$. As a result, we obtain a randomized integral solution while losing only a constant factor.
\end{proof}

\paragraph*{The \idle algorithm}
Assume that our cache has size $k+1$ and the extra slot is memoryless (as defined above). For any time step $t$, let $L(t)$ denote the set of pages predicted to arrive starting at time $t+1$. At time step 1 (i.e., initially), \idle runs the offline algorithm from Lemma~\ref{lem:const-plus1-pred} on $L(1)$, ignoring future requests. After the requests in $L(1)$ have been served, i.e., at time $|L(1)| + 1$, \idle then consults the predictor and runs the offline algorithm on the next ``batch''. The algorithm proceeds in this batch-by-batch manner until the end. We can show that the competitive ratio of this algorithm is at most a constant; the proof is nearly identical to the proof of Theorem~\ref{thm:static}, so we omit it.
 
\begin{lemma} \label{lemma:idle} \label{kplus:corelemma2}
On the prediction sequence $A$, we have $\cost(\idle) = O(1) \cdot \cost(\OPT^{+1}(A))$.
\end{lemma}

\paragraph*{The \learn algorithm}
Before defining the algorithm, we introduce another measurement of error that closely approximates $\ell_{ed}$. Recall that $A$ denotes a prediction sequence of length $m$ and $B$ denotes an input sequence of length $n$. In defining $\ell_{ed}$, two elements $A_i = B_j$ can be matched only if $\predictnext(\prev(j)) = i$, and no matching edges are permitted to cross.

\begin{definition}
The \emph{constrained edit distance} $\ell'_{ed}$ is the minimum weight of unmatched elements of $A$ and $B$, with the following additional constraint: if $|P(A_i)| \geq 2$, then $A_i$ can only be matched with the latest-arriving element in $P(A_i)$.
\end{definition}

We note that $\ell'_{ed}$ is a constant approximation of $\ell_{ed}$, as shown in the following lemma.

\begin{lemma} \label{lem:l-prime-ed}
For any sequences $A,B$, we have $\ell_{ed}\leq \ell'_{ed}\leq 3\ell_{ed}$.
\end{lemma}

\begin{proof}
The first inequality follows directly from the definitions of $\ell_{ed}$ and $\ell'_{ed}$.

Let $S=\{i: \abs{P(A_i)} \geq 2\}$, and let $w(S)=\sum_{i\in S} w_{A_i}$. Let $M$ be an optimal matching for $\ell_{ed}$. For each $i\in S$, there is at least one unmatched $B_j \in P(A_i)$ because $A_i$ can only get matched with one request in $B$. Each of these unmatched elements of $B$ contributes to the value of $\ell_{ed}$, so $\ell_{ed} \geq w(S)$.

Now we construct a feasible matching $M'$ for $\ell'_{ed}$ by removing the edges incident to $S$ from $M$, that is, $M'=\{(A_i,B_j)\in M | i\notin S \}$. Consider the requests unmatched by $M'$: the weight is at most the amount originally unmatched by $M$ together with the amount incurred from removing edges incident to $S$. The former contributes $\ell_{ed}$ weight while the latter contributes $2w(S)$ weight, so we have $\ell'_{ed}\leq \ell_{ed}+2w(S)\leq 3\ell_{ed}$.
\end{proof}

Now we are ready to define the \learn algorithm. For any $i \leq j$, we let $A(i, j)$ denote the subsequence $(A_i, A_{i+1}, \ldots, A_j)$. For any set (or multiset) of pages $S$, we let $w(S)$ denote the total cost of pages in $S$. The algorithm is the following:
\begin{enumerate}
    \item Let $s = 0$; the variable $s$ always denotes that we have imitated the \idle algorithm through the first $s$ requests of the prediction.
    \item Let $S = \emptyset$ be an empty queue.
    \item On the arrival of request $p$, add $p$ to $S$.\label{branch_steps}
    \begin{enumerate}
        \item If there is a $t$ (in $[s+1, L]$ where $L$ is the end of the current prediction) such that
        \begin{equation} \label{ineq:edit-weight}
        \ell'_{ed}(A(s+1, t), S) < \frac{1}{3}(w(A(s+1,t)) + w(S)),
        \end{equation}
        then imitate \idle through position $t$, empty $S$ and let $s=t$. (If more than one $t$ satisfies the above, select the minimum.) \label{a_large_step}
        \item Otherwise, evict the page in the final slot. \label{a_small_step}
    \end{enumerate}
\end{enumerate}

We first prove that the algorithm is indeed feasible.

\begin{lemma} \label{lem:learn-feasible}
In the \learn algorithm, Step \ref{a_large_step} is feasible, i.e., if $t$ satisfies \eqref{ineq:edit-weight}, then $A_t=p$.
\end{lemma}

\begin{proof}
Consider the optimal matching $M$ between $A(s+1,t)$ and $S$; we will show that both $A_t$ and $p$ are matched in $M$, and this implies that $(A_t, p)$ is an edge in $M$, so $A_t = p$.

For contradiction, first suppose that $A_t$ is not matched in $M$, in which case
\begin{align*}
    \ell'_{ed}(A(s+1, t-1), S) 
    &=  \ell'_{ed}(A(s+1, t), S) - w(A_t)   \\
    &<  \frac{1}{3}(w(A(s+1,t)) + w(S))-w(A_t)   \\
    &\leq  \frac{1}{3}(w(A(s+1,t-1)) + w(S)),
\end{align*}
which means $t-1$ satisfies \eqref{ineq:edit-weight}, contradicting our choice of the minimum $t$ satisfying \eqref{ineq:edit-weight}.

Now we will show that $p$ is matched in $M$. For contradiction, suppose $p$ is not matched in $M$, which means $A_t$ is matched to some other request $B_i\in S'$. By the defined matching conditions, we have $\predictnext(\prev(i))=t$. This implies that when the algorithm was serving request $B_i$, it could see the prediction sequence $A(s,t)$.

Let $S'$ denote the contents of the queue up through $B_i$, and let $w(S \setminus S')$ denote the weight of pages in $S \setminus S'$ (including $p$). Since $A_t$ is matched to $B_i$, no pages in $S \setminus S'$ can be matched when considering request $p$. Thus, we have the following:
\begin{align*}
    \ell'_{ed}(A(s+1, t), S') 
    &= \ell'_{ed}(A(s+1, t), S) - w(S \setminus S')  \\
    &< \frac{1}{3}(w(A(s+1,t)) + w(S))-w(S \setminus S')   \\
    &\leq \frac{1}{3}(w(A(s+1,t)) + w(S')),
\end{align*}
which means $S'$ satisfied \eqref{ineq:edit-weight} by matching $B_i$ with $A_t$, contradicting the fact that the algorithm did not enter Step \ref{a_large_step} at the time the queue was $S'$. 
\end{proof}

Now we arrive at the heart of the analysis: we upper bound the cost of \learn against the cost of \idle (i.e., a surrogate for $\OPT(B)$) and the constrained edit distance $\ell'_{ed}$. In particular, we prove the following lemma.

\begin{lemma} \label{kplus:corelemma3}
The algorithms \learn and \idle satisfy $\cost(\learn) \leq \cost(\idle)+12\ell'_{ed}$.
\end{lemma}

Note that the proof of Theorem~\ref{thm:1-led-upper} follows directly from Lemmas~\ref{kplus:corelemma1},~\ref{kplus:corelemma2}, and~\ref{kplus:corelemma3}.

\begin{proof}[Proof of Lemma~\ref{kplus:corelemma3}]
Let $\cost_1$ denote the total cost of Step \ref{a_large_step} and let $\cost_2$ denote the total cost of Step \ref{a_small_step}, so $\cost(\learn) = \cost_1 + \cost_2$. From the algorithm, we can see that $\cost_1 \leq \cost(\idle)$.

So now we will prove $\cost_2 \leq  12\ell'_{ed}$ by induction on the times we enter Step \ref{a_large_step}. Let $w_A(a,b)=w(A(a,b))$ and $w_B(a,b)=w(B(a,b))$. Let $\cost_2(a,b)$ denote the total cost of Step \ref{a_small_step} when it serves input requests from time $a$ to time $b$. Finally, let $\ell'_{ed}((a,b),(c,d))$ be the distance between $A[a...b]$ and $B[c...d]$ according to the definition of $\ell'_{ed}$.

If we never enter Step \ref{a_large_step}, then the algorithm trivially evicts every page of the input $B$, so
\[
\cost_2\leq 2w_B(1,n)\leq 2w_A(1,m)+2w_B(1,n) \leq 6\ell'_{ed}
\]
where the final inequality follows from the fact that we never satisfied \eqref{ineq:edit-weight}.

Now assume the $\cost_2\leq 12\ell'_{ed}$ if we enter Step \ref{a_large_step} fewer than $i$ times; we will show that $\cost_2\leq 12\ell'_{ed}$ if we enter Step \ref{a_large_step} $i$ times.

Consider the first time we enter Step \ref{a_large_step}, at which point we have read input $B(1,b)$ and we imitate $\idle$ on $A(1,a)$. From the definition of $\ell'_{ed}$, there exists some integer $c$ such that
\[
\ell'_{ed}=\ell'_{ed}((1,a),(1,c))+\ell'_{ed}((a+1,m),(c+1,n)).
\]
Consider the following cases:
\begin{enumerate}
    \item $c=b$: In this case, we have
        \begin{align*}
        \cost_2 &= \cost_2(1,b-1)+\cost_2(b+1,n)  \\
        &\leq 6\cdot \ell'_{ed}((1,a-1),(1,b-1)) + 12\cdot \ell'_{ed}((a+1,m),(b+1,n))
        \end{align*}
        where the equality holds due to Lemma~\ref{lem:learn-feasible}, and the inequality follows from the fact that we did not enter Step~\ref{a_large_step} on request $B_{b-1}$ and the induction hypothesis. Again, Lemma~\ref{lem:learn-feasible} and our choice to enter Step~\ref{a_large_step} imply that this quantity is equal to
        \[
        6\cdot \ell'_{ed}((1,a),(1,b)) + 12\cdot \ell'_{ed}((a+1,m),(b+1,n)),
        \]
        which is at most $12 \cdot \ell'_{ed}$ by the definition of $c$ and our case assumption.

    \item $c<b$: Since we did not enter Step~\ref{a_large_step} earlier, we have
        \begin{align}\label{c<b:in1}
        \ell'_{ed}((1,a_0),(1,c))\geq \frac{1}{3}\left(w_A(1,a_0)+w_B(1,c)\right)
        \end{align}
        for every $a_0 < a$. Furthermore, since we are now entering Step~\ref{a_large_step}, we have
        \begin{align}\label{c<b:in2}
        \ell'_{ed}((1,a),(1,b))\leq \frac{1}{3}\left(w_A(1,a)+w_B(1,b)\right).
        \end{align}
        Let $M$ be the optimal matching for $\ell'_{ed}((1,a),(1,b))$, and consider the following matching $M$' for $\ell'_{ed}((1,a),(1,c))$:
        \[
        M'=\{(A_i,B_j)\in M| j\leq c\}.
        \]
        Let $d_M=\sum_{(A_i,B_j)\in M}w_{A_i}-\sum_{(A_i,B_j)\in M'} w_{A_i}$, and let $a'=\arg\max_{i}(A_i,B_j)\in M'$. Now consider the constrained edit distance between $A(1,a')$ and $B(1,c)$: one option is to match $A(1,a)$ and $B(1,b)$ and remove the weight of unmatched requests. This implies the following:
        \[
        \ell'_{ed}((1,a'),(1,c))
        \leq \ell'_{ed}((1,a),(1,b))-w_A(a'+1,a)-w_B(c+1,b)+2d_M
        \]
        Rearranging the above yields
        \begin{align*}
        w_A(a'+1,a)+w_B(c+1,b)-2d_M
        &\leq \ell'_{ed}((1,a),(1,b))-\ell'_{ed}((1,a'),(1,c))    \\
        &\leq \frac{1}{3}(w_A(1,a) + w_B(1,b) - w_A(1,a') - w_B(1,c)) \\
        &= \frac{1}{3} \left(w_A(a'+1,a)+w_B(c+1,b)\right),
        \end{align*}
        where the second inequality follows from inequalities (\ref{c<b:in1}) and (\ref{c<b:in2}). Further rearranging and applying the inequality $d_M \leq w_A(a'+1,a)$ yields
        \begin{equation} \label{ineq:dm}
        d_M\geq\frac{1}{2}w_B(c+1,b).
        \end{equation}
        Now consider the optimal matching between $A(a+1,m)$ and $B(b+1,n)$. One way to form this matching is to match $A(a+1, m)$ and $B(c+1, n)$ (since $c < b$) and leave the requests matched to $B(c+1, b)$ unmatched (in addition to existing unmatched requests). The matching corresponding to $\ell'_{ed}((a+1,m),(c+1,n))$ is penalized by $d_M$ when considered as a matching for $A(a+1,m)$ and $B(b+1, m)$. Furthermore, the amount of weight in $A(a+1,m)$ matched to $B(c+1, n)$ is at most $w_B(c+1,b) - d_M$. This gives us the following:
        \begin{align}
        \ell'_{ed}((a+1,m),(b+1,n))
        &\leq \ell'_{ed}((a+1,m),(c+1,n)) - d_M + (w_B(c+1,b)-d_M) \nonumber \\
        &\leq \ell'_{ed}((a+1,m),(c+1,n)), \label{eq:b-to-c}
        \end{align}
        where the second inequality follows from \eqref{ineq:dm}. Letting $\cost(x,y)$ denote the cost incurred by the algorithm to serve $B(x,y)$, we have
        \begin{align*}
        \cost_2 &\leq \cost(1,c)+\cost(c+1,b)+\cost(b+1,n)  \\
        &\leq 2w_B(1,c)+4d_M + 12\cdot \ell'_{ed}((a+1,m),(b+1,n)) \tag{trivial upper bounds, \eqref{ineq:dm}, induction} \\
        &\leq 4(w_B(1,c)+w_A(1,a)) + 12\cdot \ell'_{ed}((a+1,m),(c+1,n)) \tag{trivial upper bounds and \eqref{eq:b-to-c}} \\
        &\leq 12\cdot \ell'_{ed}((1,a),(1,c)) + 12\cdot \ell'_{ed}((a+1,m),(c+1,n))
        \tag{we did not enter Step \ref{a_large_step} at time $c$}    \\
        &=12\cdot \ell'_{ed}.
        \end{align*}
    \item $c>b$: This case is very similar to the $c < b$ case, so we omit some details. Define $d \leq a$ such that
        \[
        \ell'_{ed} = \ell'_{ed}((1,d),(1,b)) + \ell'_{ed}((d+1,m),(b+1,n)).
        \]
        If $d = a$, then this case is analogous to the $c = b$ case, so from now on, we assume $d < a$. Then for every $b' \leq b$, we have
        \begin{align}\label{last-alg-i1}
        \ell'_{ed}((1,d),(1,b'))\geq \frac{1}{3}\left(w_A(1,d)+w_B(1,b')\right),
        \end{align}
        and since we are now entering Step \ref{a_large_step}, we have
        \begin{align}\label{last-alg-i2}
        \ell'_{ed}((1,a),(1,b))\leq \frac{1}{3}\left(w_A(1,a)+w_B(1,b)\right).
        \end{align}
        Let $M$ denote the optimal matching for $\ell'_{ed}((1,a),(1,b))$ and consider the following matching between $A(1,d)$ and $B(1,b)$:
        \[
        M'=\{(A_i,B_j)\in M| i\leq d\}.
        \]
        Let $d_M =\sum_{(A_i,B_j)\in M} w_{A_i} - \sum_{(A_i,B_j)\in M'} w_{A_i}$, and let $b'=\arg\max_{j}(A_i,B_j)\in M'$. Since $M'$ is a valid matching between $A(1,d)$ and $B(1,b')$, we have the following:
        \begin{align}\label{last-alg-i3}
        \ell'_{ed}((1,d),(1,b')) \leq \ell'_{ed}((1,a),(1,b))-w_A(d+1,a)-w_B(b'+1,b)+2d_M.
        \end{align}
        Rearranging and applying the previous inequalities \ref{last-alg-i1}, \ref{last-alg-i2}, and \ref{last-alg-i3} yields
        \begin{align*}
        w_A(d+1,a)+w_B(b'+1,b)-2d_M &\leq \ell'_{ed}((1,a),(1,b))-\ell'_{ed}((1,d),(1,b')) \\
        &\leq \frac{1}{3} \left( w_A(d+1,a)+w_B(b'+1,b) \right) ,
        \end{align*}
        and further rearranging gives us
        \[
        d_M\geq \frac{1}{3} \left( w_A(d+1,a)+w_B(b'+1,b) \right).
        \]
        Since $d_M\leq w_B(b'+1,b)$, we have
        \[
        d_M\geq \frac{1}{2}w_A(d+1,a).
        \]
        As in the previous case, we have
        \begin{align*}
        \ell'_{ed}((a+1,m),(b+1,n)) &\leq \ell'_{ed}((d+1,m),(b+1,n)) - d_M + (w_A(d+1,a)-d_M)  \\
        &\leq \ell'_{ed}((d+1,m),(b+1,n)).
        \end{align*}
        Letting $\cost(x,y)$ denote the cost of serving $B(x,y)$, we have
        \begin{align*}
        \cost_2 &\leq \cost(1,b)+\cost(b+1,n)  \\
        &\leq 2w_B(1,b) + 12\cdot \ell'_{ed}((a+1,m),(b+1,n))   \\
        &\leq 2w_B(1,b)+w_A(1,d) + 12\cdot \ell'_{ed}((d+1,m),(b+1,n))   \\
        &\leq 6\cdot \ell'_{ed}((1,d),(1,b)) + 12\cdot \ell'_{ed}((d+1,m),(b+1,n)) \\
        &=12\cdot \ell'_{ed}.
        \end{align*}
        \qedhere
\end{enumerate}
\end{proof}

\paragraph*{The \follow algorithm}
Now we show that the $\Omega(\ell_1)$ lower bound in Theorem~\ref{thm:1-l1-lower} is tight, that is, we will give an SPRP algorithm \follow that has cost $O(1) \cdot (\OPT + \ell_1)$. Recall the \static algorithm from Theorem~\ref{thm:static}. The algorithm \follow ignores its input: it simply runs \static on the prediction sequence $A$ and imitates its fetches/evictions on the input sequence $B$.

\begin{theorem}
The \follow algorithm has cost $O(1) \cdot (\OPT + \ell_1)$.
\end{theorem}

\begin{proof}
Recall from Theorem~\ref{thm:static} that $\cost(\static) \leq O(1) \cdot \OPT(A)$. Furthermore, we claim $\OPT(A) \leq \OPT(B) + 2\ell_1$. This is because on $A$, there exists an algorithm that imitates the movements of $B$: say at time $t$, $\OPT(B)$ evicts some element $b$ that had appeared in $B$ at time $v(t)$. Then $\OPT(A)$ can also evict whatever element appeared at time $v(t)$ in $A$, and if this is not $b$, then this cost can be charged to the $v(t)$ term of $\ell_1$. Each term of $\ell_1$ is charged at most twice because a specific request can be evicted and fetched at most once respectively.

By the same argument, we have $\cost(\follow) \leq \cost(\static) + 2\ell_1$. Combining these inequalities proves the theorem.
\end{proof}

\section{Conclusion} \label{sec:conclusion}
In this paper, we initiated the study of weighted paging with predictions. This continues the recent line of work in online algorithms with predictions, particularly that of Lykouris and Vassilvitski~\cite{lykouris2018competitive} on unweighted paging with predictions. We showed that unlike in unweighted paging, neither a fixed lookahead not knowledge of the next request for every page is sufficient information for an algorithm to overcome existing lower bounds in weighted paging. However, a combination of the two, which we called the strong per request prediction (SPRP) model, suffices to give a constant approximation. We also explored the question of gracefully degrading algorithms with increasing prediction error, and gave both upper and lower bounds for a set of natural measures of prediction error. The reader may note that the SPRP model is rather optimistic and requires substantial information about the future. A natural question arises: can we obtain constant competitive algorithms for weighted paging with fewer predictions? While we refuted this for the PRP and fixed lookahead models, being natural choices because they suffice for unweighted paging, it is possible that an entirely different parameterization of predictions can also yield positive results for weighted paging. We leave this as an intriguing direction for future work.


\bibliographystyle{plainurl} 
\bibliography{bibliography}

\end{document}